\documentclass[journal,onecolumn]{IEEEtran}

\usepackage{mathrsfs}
\usepackage{color}
\usepackage{enumerate}
\usepackage{mathrsfs}
\usepackage{amssymb}
\usepackage{amsfonts}
\usepackage{amsmath}
\usepackage{multirow}
\usepackage{amsthm}
\usepackage{url}
\usepackage{setspace}

\newtheorem{theorem}{Theorem}[section]

\newtheorem{lemma}[theorem]{Lemma}

\newtheorem{corollary}[theorem]{Corollary}

\newtheorem{remark}[theorem]{Remark}

\begin{document}
\begin{spacing}{1.2}
\title{Quantum Codes from Generalized Reed-Solomon Codes and Matrix-Product Codes}

\author{Tao Zhang and Gennian Ge
\thanks{The research of G. Ge was supported by the National Natural Science Foundation of China under Grant No.~61171198 and  Grant No.~11431003, the Importation and Development of High-Caliber Talents Project of Beijing Municipal Institutions, and Zhejiang Provincial Natural Science Foundation of China under Grant No.~LZ13A010001.}
\thanks{T. Zhang is with  the School of Mathematical Sciences, Capital Normal University,
Beijing 100048, China. He is also with the Department of Mathematics, Zhejiang University,
Hangzhou 310027,  China (e-mail: tzh@zju.edu.cn).}
\thanks{G. Ge is with  the School of Mathematical Sciences, Capital Normal University,
Beijing 100048, China. He is also with Beijing Center for Mathematics and Information Interdisciplinary Sciences, Beijing, 100048, China (e-mail: gnge@zju.edu.cn).}
}

\maketitle

\begin{abstract}
One of the central tasks in quantum error-correction is to construct quantum codes that have good parameters. In this paper, we construct three new classes of quantum MDS codes from classical Hermitian self-orthogonal generalized Reed-Solomon codes. We also present some classes of quantum codes from matrix-product codes. It turns out that many of our quantum codes are new in the sense that the parameters of quantum codes cannot be obtained from all previous constructions.
\end{abstract}

\section{Introduction}
Quantum error-correcting codes have attracted much attention as schemes that protect quantum states from decoherence during quantum computations and quantum communications. After the pioneering works in \cite{S95,S96}, the theory of quantum codes has developed rapidly. In \cite{CRSS97,CRSS98}, Calderbank et. al found a strong connection between a large class of quantum codes which can be seen as an analog of classical group codes, and self-orthogonal codes over $\mathbb{F}_{4}$. This was then generalized to the nonbinary case in \cite{AK01,R99}. Recently, many quantum codes have been constructed by classical linear codes with Euclidean or Hermitian self-orthogonality \cite{AKS07,CLX05,S99}.

Let $q$ be a prime power, a $q$-ary $((n,K,d))$ quantum code is a $K$-dimensional vector subspace of the Hibert space $(\mathbb{C}^{q})^{\bigotimes n}\cong\mathbb{C}^{q^{n}}$ which can detect up to $d-1$ quantum errors. Let $k=\textup{log}_{q}K$, we use $[[n,k,d]]_{q}$ to denote a $q$-ary $((n,K,d))$ quantum code. As in classical coding theory, one of the central tasks in quantum coding theory is to construct quantum codes with good parameters. The following theorem gives a bound on the achievable minimum distance of a quantum code.
\begin{theorem}$(${\rm \cite{KKK06,KL97}} Quantum Singleton Bound$)$
Quantum codes with parameters $[[n,k,d]]_{q}$ satisfy
$$2d\leq n-k+2.$$
\end{theorem}

A quantum code achieving this quantum Singleton bound is called a quantum maximum-distance-separable (MDS) code. Just as in the classical linear codes, quantum MDS codes form an important family of quantum codes. Constructing quantum MDS codes has become a central topic for quantum codes in recent years.

As we know, the length of a nontrivial $q$-ary quantum MDS codes cannot exceed $q^{2}+1$ if the classical MDS conjecture holds. The quantum MDS codes of length up to $q+1$ have been constructed for all possible dimensions \cite{GBR04} \cite{GRB04}, and many quantum MDS codes of length between $q+1$ and $q^{2}+1$ have also been obtained (see \cite{BE00,CLZ15,JLLX10,JX14,KZL14,L11,LMPZ96,LXW08} and the references therein). However, almost all known $q$-ary quantum MDS codes have minimum distance less than or equal to $\frac{q}{2}+1$.

In this paper, we construct three classes of quantum MDS codes as follows:
\begin{enumerate}
  \item[(1)] Let $q$ be an odd prime power with the form $2am+1$, then there exists a $q$-ary $[[\frac{q^{2}-1}{a},\frac{q^{2}-1}{a}-2d+2,d]]$-quantum MDS code, where $2\leq d\leq (a+1)m+1$.
  \item[(2)] Let $q$ be an odd prime power with the form $2am-1$, then there exists a $q$-ary $[[\frac{q^{2}-1}{2a}-q+1,\frac{q^{2}-1}{2a}-q-2d+3,d]]$-quantum MDS code, where $2\leq d\leq (a+1)m-2$.
  \item[(3)] Let $q$ be an odd prime power with the form $(2a+1)m-1$, then there exists a $q$-ary $[[\frac{q^{2}-1}{2a+1}-q+1,\frac{q^{2}-1}{2a+1}-q-2d+3,d]]$-quantum MDS code, where $2\leq d\leq (a+1)m-1$.
\end{enumerate}
In addition, we also show the existence of $q$-ary quantum MDS codes with length $q^{2}-1$ and minimum distance $d$ for any $2\leq d\leq q$, where $q$ is an odd prime power. This result extends those given in \cite{GBR04,LXW08}.

Matrix-product codes were introduced in \cite{BN01} as a generalization of several well known constructions of longer codes from old ones, for example, the $(a|a+b)$-construction and the $(a+x|b+x|a+b+x)$-construction. In \cite{GHR14}, the authors construct some new quantum codes from matrix-product codes and the Euclidean construction. Motivated by this work, we will give a new construction of quantum codes by matrix-product codes and the Hermitian construction. Some of them have better parameters than the quantum codes listed in table online \cite{E12}.

This paper is organized as follows. In Section~\ref{pre} we recall the basics about linear codes, quantum codes and matrix-product codes. In Section~\ref{quanMDS}, we give three new classes of quantum MDS codes from generalized Reed-Solomon codes. In Section~\ref{QuanMP}, we present a new construction of quantum codes via matrix-product codes and the Hermitian construction.
\section{Preliminaries}\label{pre}
Throughout this paper, let $\mathbb{F}_{q}$ be the finite field with $q$ elements, where $q$ is a prime power. A linear $[n,k]$ code $C$ over $\mathbb{F}_{q}$ is a $k$-dimensional subspace of $\mathbb{F}_{q}^{n}$. The weight $\textup{wt}(x)$ of a codeword $x\in C$ is the number of nonzero components of $x$. The distance of two codewords $x,y\in C$ is $d(x,y)=\textup{wt}(x-y)$. The minimum distance $d$ of $C$ is the minimum distance between any two distinct codewords of $C$. An $[n,k,d]$ code is an $[n,k]$ code with the minimum distance $d$.

Given two vectors $x=(x_{0},x_{1},\cdots,x_{n-1}),\ y=(y_{0},y_{1},\cdots,y_{n-1})\in\mathbb{F}_{q}^{n}$, there are two inner products we are interested in. One is the Euclidean inner product which is defined as $\langle x,y\rangle_{E}=\sum_{i=0}^{n-1}x_{i}y_{i}.$ When $q=l^{2}$, where $l$ is a prime power, then we can also consider the Hermitian inner product which is defined by $\langle x,y\rangle_{H}=\sum_{i=0}^{n-1}x_{i}y_{i}^{l}.$
The Euclidean dual code of $C$ is defined as
$$C^{\bot E}=\{x\in\mathbb{F}_{q}^{n}|\langle x,y\rangle_{E}=0\textup{ for all }y\in C\}.$$
Similarly the Hermitian dual code of $C$ is defined as
$$C^{\bot H}=\{x\in\mathbb{F}_{q}^{n}|\langle x,y\rangle_{H}=0\textup{ for all }y\in C\}.$$
A linear code $C$ is called Euclidean (Hermitian) self-orthogonal if $C\subseteq C^{\bot E}$ ($C\subseteq C^{\bot H}$, respectively), and $C$ is called Euclidean (Hermitian) dual containing if $C^{\bot E}\subseteq C$ ($C^{\bot H}\subseteq C$, respectively).

For a vector $x=(x_{1},\cdots,x_{n})\in\mathbb{F}_{q^{2}}^{n}$, let $x^{q}=(x_{1}^{q},\cdots,x_{n}^{q})$. For a subset $S$ of $\mathbb{F}_{q^{2}}^{n}$, we define $S^{q}$ to be the set $\{x^{q}|x\in S\}$. Then it is easy to see that for a $q^{2}$-ary linear code $C$, we have $C^{\bot H}=(C^{q})^{\bot E}$. Therefore, $C$ is Hermitian self-orthogonal if and only if $C\subseteq (C^{q})^{\bot E}$, i.e., $C^{q}\subseteq C^{\bot E}$.
\subsection{Quantum Codes}
In this subsection, we recall the basics of quantum codes. Let $q$ be a power of a prime number $p$. A qubit $|v\rangle$ is a nonzero vector in $\mathbb{C}^{q}$ which can be represented as $|v\rangle=\sum_{x\in\mathbb{F}_{q}}c_{x}|x\rangle$, where $\{|x\rangle|x\in\mathbb{F}_{q}\}$ is a basis of $\mathbb{C}^{q}$. For $n\geq1$, the $n$-th tensor product $(\mathbb{C}^{q})^{\bigotimes n}\cong\mathbb{C}^{q^{n}}$ has a basis $\{|a_{1}\cdots a_{n}\rangle=|a_{1}\rangle\bigotimes\cdots\bigotimes|a_{n}\rangle|(a_{1},\cdots,a_{n})\in\mathbb{F}_{q}^{n}\},$ then an $n$-qubit is a nonzero vector in $\mathbb{C}^{q^{n}}$ which can be represented as $|v\rangle=\sum_{a\in\mathbb{F}_{q}^{n}}c_{a}|a\rangle,$ where $c_{a}\in \mathbb{C}.$

Let $\zeta_{p}$ be a primitive $p$-th root of unity. The quantum errors in $q$-ary quantum system are linear operators acting on $\mathbb{C}^{q}$ and can be represented by the set of error bases: $\varepsilon_{n}=\{T^{a}R^{b}|a,b\in\mathbb{F}_{q}\}$, where $T^{a}R^{b}$ is defined by
$$T^{a}R^{b}|x\rangle=\zeta_{p}^{\textup{Tr}_{\mathbb{F}_{q}/\mathbb{F}_{p}}(bx)}|x+a\rangle.$$
The set
$$E_{n}=\{\zeta_{p}^{l}T^{a}R^{b}|0\leq l\leq p-1,a=(a_{1},\cdots,a_{n}),b=(b_{1},\cdots,b_{n})\in\mathbb{F}_{q}^{n}\}$$
forms an error group, where $\zeta_{p}^{l}T^{a}R^{b}$ is defined by
$$\zeta_{p}^{l}T^{a}R^{b}|x\rangle=\zeta_{p}^{l}T^{a_{1}}R^{b_{1}}|x_{1}\rangle\bigotimes\cdots\bigotimes T^{a_{n}}R^{b_{n}}|x_{n}\rangle=\zeta_{p}^{l+\textup{Tr}_{\mathbb{F}_{q}/\mathbb{F}_{p}}(bx)}|x+a\rangle,$$
for any $|x\rangle=|x_{1}\rangle\bigotimes\cdots\bigotimes|x_{n}\rangle$, $x=(x_{1},\cdots,x_{n})\in\mathbb{F}_{q}^{n}$. For an error $e=\zeta_{p}^{l}T^{a}R^{b}$, its quantum weight is defined by
$$w_{Q}(e)=\sharp\{1\leq i\leq n|(a_{i},b_{i})\neq(0,0)\}.$$

A subspace $Q$ of $\mathbb{C}^{q^{n}}$ is called a $q$-ary quantum code with length $n$. The $q$-ary quantum code has minimum distance $d$ if and only if it can detect all errors in $E_{n}$ of quantum weight less than $d$, but cannot detect some errors of weight $d$. A $q$-ary $[[n,k,d]]_{q}$ quantum code is a $q^{k}$-dimensional subspace of $\mathbb{C}^{q^{n}}$ with minimum distance $d$. There are many methods to construct quantum codes, and the following theorem is one of the most frequently used construction methods.
\begin{theorem}$($\rm{\cite{AK01}} Hermitian Construction$)$\label{thm1}
If $C$ is a $q^{2}$-ary Hermitian dual-containing $[n,k,d]$ code, then there exists a $q$-ary $[[n,2k-n,\geq d]]$-quantum code.
\end{theorem}

Then we have the following corollary.
\begin{corollary}\label{HerMDS}
There is a $q$-ary $[[n,n-2k,k+1]]$ quantum MDS code whenever there exists a $q^{2}$-ary classical Hermitian self-orthogonal $[n,k,n-k+1]$-MDS code.
\end{corollary}
\subsection{Matrix-Product Codes}
In this subsection, we review some notations and results of matrix-product codes. Let $C_{1},C_{2},\cdots,C_{s}$ be a family of $s$ codes of length $m$ over $\mathbb{F}_{q}$ and $A=(a_{ij})$ be an $s\times l$ matrix with entries in $\mathbb{F}_{q}$. Then, the matrix-product code $[C_{1},C_{2},\cdots,C_{s}]\cdot A$ is defined as the code over $\mathbb{F}_{q}$ of length $ml$ with generator matrix
$$\left(
  \begin{array}{cccc}
    a_{11}G_{1} & a_{12}G_{1} & \cdots & a_{1l}G_{1} \\
    a_{21}G_{2} & a_{22}G_{2} & \cdots & a_{2l}G_{2} \\
    \vdots & \vdots & \cdots & \vdots \\
    a_{s1}G_{s} & a_{s2}G_{s} & \cdots & a_{sl}G_{s} \\
  \end{array}
\right),$$
where $G_{i},\ 1\leq i\leq s$, is a generator matrix for the code $C_{i}$. The following theorem gives a characterization of matrix-product codes.
\begin{theorem}\label{MP3}\rm{\cite{HLR09,GHR14}}
The matrix-product code $[C_{1},C_{2},\cdots,C_{s}]\cdot A$ given by a sequence of $[m,k_{i},d_{i}]$ linear codes $C_{i}$ over $\mathbb{F}_{q}$ and a full-rank $s\times l$ matrix $A$ is a linear code whose length is $ml$, it has dimension $\sum_{i=1}^{s}k_{i}$ and minimum distance larger than or equal to $\delta=\textup{min}_{1\leq i\leq s}\{d_{i}\delta_{i}\}$, where $\delta_{i}$ is the minimum distance of the code on $\mathbb{F}_{q}^{l}$ generated by the first $i$ rows of the matrix $A$.
\end{theorem}

In order to construct quantum codes from matrix-product codes, we need the following theorem.
\begin{theorem}\label{MP1}\rm{\cite{BN01}}
Assume that $C_{1},C_{2},\cdots,C_{s}$ are a family of linear codes of length $m$ and $A$ is a nonsingular $s\times s$ matrix, then the following equality of codes happens
$$([C_{1},C_{2},\cdots,C_{s}]\cdot A)^{\bot}=[C_{1}^{\bot},C_{2}^{\bot},\cdots,C_{s}^{\bot}]\cdot (A^{-1})^{t},$$
where $B^{t}$ denotes the transpose of the matrix $B$.
\end{theorem}
\section{New Quantum MDS Codes from Generalized Reed-Solomon Codes}\label{quanMDS}
We first recall the basics of generalized Reed-Solomon codes. Choose $n$ distinct elements $a_{1},\cdots,a_{n}$ of $\mathbb{F}_{q}$ and $n$ nonzero elements $v_{1},\cdots,v_{n}$ of  $\mathbb{F}_{q}$. For $1\leq k\leq n$, we define the code
$$\textup{GRS}_{k}(\mathbf{a},\mathbf{v}):=\{(v_{1}f(a_{1}),\cdots,v_{n}f(a_{n}))|f(x)\in\mathbb{F}_{q}[x]\textup{ and deg}(f(x))<k\},$$
where $\mathbf{a}$ and $\mathbf{v}$ denote the vectors $(a_{1},\cdots,a_{n})$ and $(v_{1},\cdots,v_{n})$, respectively. The code $\textup{GRS}_{k}(\mathbf{a},\mathbf{v})$ is called a generalized Reed-Solomon code over $\mathbb{F}_{q}$. It is well known that a generalized Reed-Solomon code $\textup{GRS}_{k}(\mathbf{a},\mathbf{v})$ is an MDS code with parameters $[n,k,n-k+1]$. The following lemma presents a criterion to determine whether or not a generalized Reed-Solomon code is Hermian self-orthogonal. We assume that $0^{0}:=1$.
\begin{lemma}\label{lemGRS}
Let $\mathbf{a}=(a_{1},\cdots,a_{n})\in\mathbb{F}_{q^{2}}^{n}$ and $\mathbf{v}=(v_{1},\cdots,v_{n})\in(\mathbb{F}_{q^{2}}^{*})^{n}$, then $\textup{GRS}_{k}(\mathbf{a},\mathbf{v})\subseteq\textup{GRS}_{k}(\mathbf{a},\mathbf{v})^{\bot H}$ if and only if $\sum_{i=1}^{n}v_{i}^{q+1}a_{i}^{qj+l}=0$ for all $0\leq j,l\leq k-1$.
\end{lemma}
\begin{proof}
Note that $\textup{GRS}_{k}(\mathbf{a},\mathbf{v})\subseteq\textup{GRS}_{k}(\mathbf{a},\mathbf{v})^{\bot H}$ if and only if $\textup{GRS}_{k}(\mathbf{a},\mathbf{v})^{q}\subseteq\textup{GRS}_{k}(\mathbf{a},\mathbf{v})^{\bot E}$. It is obvious that $\textup{GRS}_{k}(\mathbf{a},\mathbf{v})^{q}$ has a basis $\{(v_{1}^{q}a_{1}^{iq},\cdots,v_{n}^{q}a_{n}^{iq})|0\leq i\leq k-1\}$, and $\textup{GRS}_{k}(\mathbf{a},\mathbf{v})$ has a basis $\{(v_{1}a_{1}^{i},\cdots,v_{n}^{}a_{n}^{i})|0\leq i\leq k-1\}$. So $\textup{GRS}_{k}(\mathbf{a},\mathbf{v})^{q}\subseteq\textup{GRS}_{k}(\mathbf{a},\mathbf{v})^{\bot E}$ if and only if $\sum_{i=1}^{n}v_{i}^{q+1}a_{i}^{qj+l}=0$ for all $0\leq j,l\leq k-1$.
\end{proof}
Now we consider generalized Reed-Solomon codes over $\mathbb{F}_{q^{2}}$ to construct quantum codes.
\subsection{New Quantum MDS Codes of Length $\frac{q^{2}-1}{a}$}
\begin{theorem}\label{thmNQ2}
Let $q$ be an odd prime with the form $2am+1$, $\omega$ be a fixed primitive element of $\mathbb{F}_{q^{2}}$ and $n=\frac{q^{2}-1}{a}$. Suppose $\mathbf{a}=(\omega^{a},\omega^{2a},\cdots,\omega^{na})\in\mathbb{F}_{q^{2}}^{n}$, $\mathbf{v}=(\omega^{q-1},\omega^{q-1},\omega^{a},\omega^{a},\cdots,\omega^{q-1-a},\omega^{q-1-a},\cdots,\omega^{q-1},\omega^{q-1},\omega^{a},\omega^{a},\cdots,\omega^{q-1-a},\omega^{q-1-a})\in\mathbb{F}_{q^{2}}^{n}$ and $1\leq k\leq (a+1)m$. Then $\textup{GRS}_{k}(\mathbf{a},\mathbf{v})\subseteq\textup{GRS}_{k}(\mathbf{a},\mathbf{v})^{\bot H}$.
\end{theorem}
\begin{proof}
For $0\leq j,l\leq k-1\leq (a+1)m-1$, we have
 \begin{align*}
&\sum_{i=1}^{n}v_{i}^{q+1}a_{i}^{qj+l} \\
&=\sum_{i=1}^{\frac{q-1}{a}}\omega^{ai(q+1)}\sum_{s=0}^{\frac{q-1}{2}}(\omega^{[2s(q-1)+a(2i+1)](qj+l)}+\omega^{[2s(q-1)+a(2i+2)](qj+l)})\\
&=\sum_{i=1}^{\frac{q-1}{a}}\omega^{ai(q+1)}(\omega^{a(2i+1)(qj+l)}+\omega^{a(2i+2)(qj+l)})\sum_{s=0}^{\frac{q-1}{2}}\omega^{2s(q-1)(qj+l)}.
\end{align*}
Note that
\[\sum_{s=0}^{\frac{q-1}{2}}\omega^{2s(q-1)(qj+l)}=\begin{cases}0;&\textup{ if }\frac{q+1}{2}\nmid(qj+l),\\
\frac{q+1}{2};&\textup{ if }\frac{q+1}{2}|(qj+l).\end{cases}\]

Now assume that $qj+l=t\frac{q+1}{2}$. We claim $\frac{q-1}{a}\nmid(t+1)$, otherwise $t=r\frac{q-1}{a}-1$ for $1\leq r\leq a$ since $0\leq j,l\leq q-2$. But $qj+l=(r\frac{q-1}{a}-1)\frac{q+1}{2}=(r\frac{q-1}{2a}-1)q+r\frac{q-1}{2a}+\frac{q-1}{2}$, then $(a+1)m\leq l=r\frac{q-1}{2a}+\frac{q-1}{2}\leq q-1$, which is a contradiction. Thus
 \begin{align*}
&\sum_{i=1}^{n}v_{i}^{q+1}a_{i}^{qj+l} \\
&=\frac{q+1}{2}\sum_{i=1}^{\frac{q-1}{a}}\omega^{ai(q+1)}(\omega^{a(2i+1)(qj+l)}+\omega^{a(2i+2)(qj+l)})\\
&=\frac{q+1}{2}\omega^{at(q+1)}\sum_{i=1}^{\frac{q-1}{a}}\omega^{ai(q+1)(t+1)}+\frac{q+1}{2}\omega^{at\frac{q+1}{2}}\sum_{i=1}^{\frac{q-1}{a}}\omega^{ai(q+1)(t+1)}\\
&=0.
\end{align*}
Then by Lemma~\ref{lemGRS}, $\textup{GRS}_{k}(\mathbf{a},\mathbf{v})\subseteq\textup{GRS}_{k}(\mathbf{a},\mathbf{v})^{\bot H}$.
\end{proof}

\begin{theorem}\label{Mainthm}
Let $q$ be an odd prime power with the form $2am+1$, then there exists a $q$-ary $[[\frac{q^{2}-1}{a},\frac{q^{2}-1}{a}-2d+2,d]]$-quantum MDS code, where $2\leq d\leq (a+1)m+1$.
\end{theorem}
\begin{proof}
The proof is a straightforward application of Corollary~\ref{HerMDS} and Theorem~\ref{thmNQ2}.
\end{proof}

As an immediate consequence of Theorem~\ref{Mainthm}, we have the following corollary by taking $a=1$.
\begin{corollary}\label{quanMDS2}
Let $q$ be an odd prime power, then there exists a $q$-ary $[[q^{2}-1,q^{2}-2d+1,d]]$-quantum MDS code, where $2\leq d\leq q$.
\end{corollary}
\begin{remark}
In \cite{GBR04,LXW08}, the authors showed that there exists a $q$-ary $[[q^{2}-1,q^{2}-2d+1,d]]$-quantum MDS code, where $q$ is an odd prime power and $2\leq d\leq q-1$. Obviously, our result has larger minimum distance.
\end{remark}
\subsection{New Quantum MDS Codes of Length $\frac{q^{2}-1}{2a}-q+1$}
In this subsection, we will construct quantum MDS codes of length $\frac{q^{2}-1}{2a}-q+1$. For our purpose, we need the following lemma.
\begin{lemma}\rm{\cite{JX14}}\label{lemNQ}
Let $A$ be an $(n-1)\times n$ matrix of rank $n-1$ over $\mathbb{F}_{q^{2}}$. Then the equation $Ax=0$ has a nonzero solution in $\mathbb{F}_{q}$ if and only if $A^{(q)}$ and $A$ are row equivalent, where $A^{(q)}$ is obtained from $A$ by raising every entry to its $q$-th power.
\end{lemma}
\begin{theorem}\label{thmNQ3}
Let $q$ be an odd prime power with the form $2am-1$ and $n=\frac{q^{2}-1}{2a}-q+1$, then there exist $\bf{a}\in\mathbb{F}_{q^{2}}^{n}$ and $\bf{v}\in(\mathbb{F}_{q^{2}}^{*})^{n}$ such that $\textup{GRS}_{k}(\mathbf{a},\mathbf{v})\subseteq\textup{GRS}_{k}(\mathbf{a},\mathbf{v})^{\bot H}$ for $1\leq k\leq (a+1)m-3$.
\end{theorem}
\begin{proof}
Let $\omega$ be a fixed primitive element of $\mathbb{F}_{q^{2}}$. We also let $A$ be an $(m-2)\times(m-1)$ matrix with $A_{ij}=\omega^{2j(m-3+(q-1)(i-1))}$ for $1\leq i\leq m-2,\ 1\leq j\leq m-1$.

Since $(m-3+(q-1)(i-1))q\equiv(m-3+(q-1)(m-i-2))\pmod{q^{2}-1}$ for $1\leq i\leq m-2$, then $A^{(q)}$ and $A$ are row equivalent. By Lemma~\ref{lemNQ}, there exists $\mathbf{c}\in\mathbb{F}_{q}^{m-1}$ such that $A\cdot \mathbf{c}^{t}=0$. Note that by deleting any one column of matrix $A$, the remaining matrix is a Vandermonde matrix, hence all coordinates of $\mathbf{c}$ are nonzero. So we can represent $\mathbf{c}$ as $\mathbf{c}=(\omega^{a_{1}(q+1)},\cdots,\omega^{a_{m-1}(q+1)})$.

Now let $\mathbf{a}=(\omega^{2a},\omega^{4a},\cdots,\omega^{q+1-2a},\omega^{q+1+2a},\cdots,\omega^{2q+2-2a},\cdots,\omega^{q^{2}-q-2+2a},\cdots,\omega^{q^{2}-1-2a})\in\mathbb{F}_{q^{2}}^{n}$ and $\mathbf{v}=(\omega^{a_{1}},\cdots,\\ \omega^{a_{m-1}},\omega^{a_{1}-(m-3)},\cdots,\omega^{a_{m-1}-(m-3)},\cdots,\omega^{a_{1}-(m-3)(q-2)},\cdots,\omega^{a_{m-1}-(m-3)(q-2)})\in(\mathbb{F}_{q^{2}}^{*})^{n}$.
Then for $0\leq j,l\leq k-1\leq (a+1)m-4$, we have
 \begin{align*}
&\sum_{i=1}^{n}v_{i}^{q+1}a_{i}^{qj+l} \\
&=\sum_{i=1}^{m-1}\omega^{a_{i}(q+1)+2ai(qj+l)}\sum_{s=0}^{q-2}\omega^{(q+1)(qj+l-m+3)s}.
\end{align*}
Note that
\[\sum_{s=0}^{q-2}\omega^{(q+1)(qj+l-m+3)s}=\begin{cases}0;&\textup{ if }(q-1)\nmid(qj+l-m+3),\\
q-1;&\textup{ if }(q-1)|(qj+l-m+3).\end{cases}\]

Assume $qj+l-m+3=t(q-1)$, we claim that $t\not\equiv m-2,m-1\pmod{m}$. Otherwise, if $t\equiv m-2\pmod{m}$, let $t=rm+m-2$, then $0\leq r\leq a$. If $r\leq a-1$, then $qj+l=t(q-1)+m-3=(mr+m-3)q+(q-mr-1)=(mr+m-3)q+(2a-r)m-2$ and $(2a-r)m-2>(a+1)m-4$ which is a contradiction. If $r=a$, then $qj+l=t(q-1)+m-3=(am+m-3)q+(am-2)$ and $am+m-3>(a+1)m-4$, which is also a contradiction. Similarly, $t\not\equiv m-1\pmod{m}$. Hence $qj+l\pmod{\frac{q^{2}-1}{2a}}\in\{t(q-1)+m-3|0\leq t\leq m-3\}$. Thus
 \begin{align*}
&\sum_{i=1}^{n}v_{i}^{q+1}a_{i}^{qj+l} \\
&=(q-1)\sum_{i=1}^{m-1}\omega^{a_{i}(q+1)+2ai(qj+l)}\\
&=0,
\end{align*}
where the last equation is from the definition of $\mathbf{c}$. Then by Lemma~\ref{lemGRS}, $\textup{GRS}_{k}(\mathbf{a},\mathbf{v})\subseteq\textup{GRS}_{k}(\mathbf{a},\mathbf{v})^{\bot H}$.
\end{proof}

\begin{theorem}
Let $q$ be an odd prime power with the form $2am-1$, then there exists a $q$-ary $[[\frac{q^{2}-1}{2a}-q+1,\frac{q^{2}-1}{2a}-q-2d+3,d]]$-quantum MDS code, where $2\leq d\leq (a+1)m-2$.
\end{theorem}
\begin{proof}
The proof is a straightforward application of Corollary~\ref{HerMDS} and Theorem~\ref{thmNQ3}.
\end{proof}

In particular, taking $a=1$, we obtain the following corollary.
\begin{corollary}
Let $q\geq5$ be an odd prime power, then there exists a $q$-ary $[[\frac{q^{2}-1}{2}-q+1,\frac{q^{2}-1}{2}-q-2d+3,d]]$-quantum MDS code, where $2\leq d\leq q-1$.
\end{corollary}
\subsection{New Quantum MDS Codes of Length $\frac{q^{2}-1}{2a+1}-q+1$}
In this subsection, we consider quantum MDS codes of length $\frac{q^{2}-1}{2a+1}-q+1$.
\begin{theorem}\label{thmNQ4}
Let $q$ be an odd prime power with the form $(2a+1)m-1$ and $n=\frac{q^{2}-1}{2a+1}-q+1$, then there exist $\bf{a}\in\mathbb{F}_{q^{2}}^{n}$ and $\bf{v}\in(\mathbb{F}_{q^{2}}^{*})^{n}$ such that $\textup{GRS}_{k}(\mathbf{a},\mathbf{v})\subseteq\textup{GRS}_{k}(\mathbf{a},\mathbf{v})^{\bot H}$ for $1\leq k\leq (a+1)m-2$.
\end{theorem}
\begin{proof}
Let $\omega$ be a fixed primitive element of $\mathbb{F}_{q^{2}}$. We also let $A$ be an $(m-2)\times (m-1)$ matrix with $A_{ij}=\omega^{(i(q-1)-1)j}$ for $1\leq i\leq m-2,\ 1\leq j\leq m-1$.

Since $(i(q-1)-1)q\equiv((q-i)(q-1)-1)\pmod{q^{2}-1}$ for $1\leq i\leq q-1$, then $A^{(q)}$ and $A$ are row equivalent. By Lemma~\ref{lemNQ}, there exists $\mathbf{c}\in\mathbb{F}_{q}^{m-1}$ such that $A\cdot \mathbf{c}^{t}=0$. Since by deleting any one column of matrix $A$, the remaining matrix is a Vandermonde matrix, then all coordinates of $\mathbf{c}$ are nonzero. Hence we can represent $\mathbf{c}$ as $\mathbf{c}=(\omega^{a_{1}(q+1)},\cdots,\omega^{a_{m-1}(q+1)})$.

Now let $\mathbf{a}=(\omega^{2a+1},\omega^{2(2a+1)},\cdots,\omega^{q-2a},\omega^{q+2a+2},\cdots,\omega^{2q+1-2a},\cdots,\omega^{q^{2}-q-1+2a},\cdots,\omega^{q^{2}-2-2a})\in\mathbb{F}_{q^{2}}^{n}$ and $\mathbf{v}=(\omega^{a_{1}},\cdots,\omega^{a_{m-1}},\omega^{a_{1}+1},\cdots, \omega^{a_{m-1}+1}, \cdots,\omega^{a_{1}+q-2},\cdots,\omega^{a_{m-1}+q-2})\in(\mathbb{F}_{q^{2}}^{*})^{n}$.
Then for $0\leq j,l\leq k-1\leq (a+1)m-3$, we have
 \begin{align*}
&\sum_{i=1}^{n}v_{i}^{q+1}a_{i}^{qj+l} \\
&=\sum_{i=1}^{m-1}\omega^{a_{i}(q+1)+i(qj+l)}\sum_{s=0}^{q-2}\omega^{(q+1)(qj+l+1)s}.
\end{align*}
Note that
\[\sum_{s=0}^{q-2}\omega^{(q+1)(qj+l+1)s}=\begin{cases}0;&\textup{ if }(q-1)\nmid(qj+l+1),\\
q-1;&\textup{ if }(q-1)|(qj+l+1).\end{cases}\]

Assume $qj+l+1=t(q-1)$, then $t\not\equiv 0,m-1\pmod{m}$. Otherwise, if $t\equiv0\pmod{m}$, let $t=rm$. Then $qj+l=t(q-1)-1=(rm-1)q+(2a+1-r)m-2$ and $\textup{min}\{rm-1,(2a+1-r)m-2\}>(a+1)m-3$, which is a contradiction. Similarly, $t\not\equiv m-1\pmod{m}$. Thus
 \begin{align*}
&\sum_{i=1}^{n}v_{i}^{q+1}a_{i}^{qj+l} \\
&=(q-1)\sum_{i=1}^{m-1}\omega^{a_{i}(q+1)+i(qj+l)}\\
&=0,
\end{align*}
where the last equation is from the definition of $\mathbf{c}$. Then by Lemma~\ref{lemGRS}, $\textup{GRS}_{k}(\mathbf{a},\mathbf{v})\subseteq\textup{GRS}_{k}(\mathbf{a},\mathbf{v})^{\bot H}$.
\end{proof}

\begin{theorem}
Let $q$ be an odd prime power with the form $(2a+1)m-1$, then there exists a $q$-ary $[[\frac{q^{2}-1}{2a+1}-q+1,\frac{q^{2}-1}{2a+1}-q-2d+3,d]]$-quantum MDS code, where $2\leq d\leq (a+1)m-1$.
\end{theorem}
\begin{proof}
The proof is a straightforward application of Corollary~\ref{HerMDS} and Theorem~\ref{thmNQ4}.
\end{proof}

In particular, taking $a=0$, we obtain the following corollary.
\begin{corollary}
Let $q$ be an odd prime power, then there exists a $q$-ary $[[q^{2}-q,q^{2}-q-2d+2,d]]$-quantum MDS code, where $2\leq d\leq q$.
\end{corollary}
\section{New Quantum Codes from Matrix-Product Codes}\label{QuanMP}
Let $A=(a_{ij})$ be an $s\times s$ matrix with entries in $\mathbb{F}_{q^{2}}$, we define $A^{(q)}=(a_{ij}^{q})$. Then we have the following result.
\begin{lemma}\label{MP2}
Let $A=(a_{ij})$ be a nonsingular $s\times s$ matrix such that $A^{(q)}$ is also nonsingular. Suppose there exist linear codes $C_{i}$ such that $C_{i}^{\bot H}\subseteq C_{i}$ for $i=1,2,\cdots,s$. Then
$$([C_{1},C_{2},\cdots,C_{s}]\cdot A)^{\bot H}\subseteq[C_{1},C_{2},\cdots,C_{s}]\cdot [(A^{(q)})^{-1}]^{t}.$$
\end{lemma}
\begin{proof}
By Theorem~\ref{MP1}, we have
 \begin{align*}
([C_{1},C_{2},\cdots,C_{s}]\cdot A)^{\bot H}&=([C_{1}^{q},C_{2}^{q},\cdots,C_{s}^{q}]\cdot A^{(q)})^{\bot E} \\
&=[(C_{1}^{q})^{\bot E},(C_{2}^{q})^{\bot E},\cdots,(C_{s}^{q})^{\bot E}]\cdot [(A^{(q)})^{-1}]^{t}\\
&=[C_{1}^{\bot H},C_{2}^{\bot H},\cdots,C_{s}^{\bot H}]\cdot [(A^{(q)})^{-1}]^{t}\\
&\subseteq[C_{1},C_{2},\cdots,C_{s}]\cdot [(A^{(q)})^{-1}]^{t}.
\end{align*}
\end{proof}

Now we would like to use matrix-product codes to construct quantum codes.
\begin{corollary}\label{MP4}
Let $q=p^{t}$ be an odd prime power, where $p$ is a prime number. Suppose $C_{1},C_{2}$ are linear codes over $\mathbb{F}_{q^{2}}$ with parameters $[n,k_{i},d_{i}]$ and $C_{i}^{\bot H}\subseteq C_{i}$, $i=1,2$. Then there exists a Hermitian dual containing $[2n,k_{1}+k_{2},\geq\textup{min}\{2d_{1},d_{2}\}]$ code over $\mathbb{F}_{q^{2}}$.
\end{corollary}
\begin{proof}
Take
$$A=\left(
  \begin{array}{cc}
    1 & 1 \\
    1 & p-1 \\
  \end{array}
\right),$$
then
$$[(A^{(q)})^{-1}]^{t}=\left(
  \begin{array}{cc}
    \frac{p+1}{2} & \frac{p+1}{2} \\
    \frac{p+1}{2} & \frac{p-1}{2} \\
  \end{array}
\right).$$
By Lemma~\ref{MP2}, we have
 \begin{align*}
([C_{1},C_{2}]\cdot A)^{\bot H}&\subseteq[C_{1},C_{2}]\cdot [(A^{(q)})^{-1}]^{t} \\
&=[C_{1},C_{2}]\cdot A.
\end{align*}
Applying Theorem~\ref{MP3}, $[C_{1},C_{2}]\cdot A$ is a Hermitian dual containing  $[2n,k_{1}+k_{2},\geq\textup{min}\{2d_{1},d_{2}\}]$ code.
\end{proof}

The following result can be found in \cite{GBR04,JLLX10,JX14}.
\begin{theorem}\label{MP5}
Let $q$ be an odd prime power, then
\begin{enumerate}
  \item there exists a $q^{2}$-ary Hermitian dual containing $[q^{2}+1,q^{2}+2-d,d]$ code for $1\leq d\leq q+1$;
  \item there exists a $q^{2}$-ary Hermitian dual containing $[q^{2},q^{2}+1-d,d]$ code for $2\leq d\leq q$.
\end{enumerate}
\end{theorem}

By combing Theorems~\ref{MP5},~\ref{thmNQ2} and  Corollary~\ref{MP4}, we can immediately get the following lemma.
\begin{lemma}\label{MP6}
Let $q$ be an odd prime power, then
\begin{enumerate}
  \item there exists a $q^{2}$-ary Hermitian dual containing $[2q^{2}+2,2q^{2}+4-d-\frac{d}{2},d]$ code, where $2\leq d\leq q+1$ is even;
  \item there exists a $q^{2}$-ary Hermitian dual containing $[2q^{2}+2,2q^{2}+3-d-\frac{d-1}{2},d]$ code, where $2\leq d\leq q+1$ is odd;
  \item there exists a $q^{2}$-ary Hermitian dual containing $[2q^{2},2q^{2}+2-d-\frac{d}{2},d]$ code, where $2\leq d\leq q$ is even;
  \item there exists a $q^{2}$-ary Hermitian dual containing $[2q^{2},2q^{2}+1-d-\frac{d-1}{2},d]$ code, where $2\leq d\leq q$ is odd;
  \item there exists a $q^{2}$-ary Hermitian dual containing $[2q^{2}-2,2q^{2}-d-\frac{d}{2},d]$ code, where $2\leq d\leq q$ is even;
  \item there exists a $q^{2}$-ary Hermitian dual containing $[2q^{2}-2,2q^{2}-1-d-\frac{d-1}{2},d]$ code, where $2\leq d\leq q$ is odd.
\end{enumerate}
\end{lemma}

Then by the Hermitian construction and Lemma~\ref{MP6}, we have the following theorem.
\begin{theorem}\label{MP7}
Let $q$ be an odd prime power, then
\begin{enumerate}
  \item there exists a $q$-ary $[[2q^{2}+2,2q^{2}+6-3d,\geq d]]$ quantum code, where $2\leq d\leq q+1$ is even;
  \item there exists a $q$-ary $[[2q^{2}+2,2q^{2}+5-3d,\geq d]]$ quantum code, where $2\leq d\leq q+1$ is odd;
  \item there exists a $q$-ary $[[2q^{2},2q^{2}+4-3d,\geq d]]$ quantum code, where $2\leq d\leq q$ is even;
  \item there exists a $q$-ary $[[2q^{2},2q^{2}+3-3d,\geq d]]$ quantum code, where $2\leq d\leq q$ is odd;
  \item there exists a $q$-ary $[[2q^{2}-2,2q^{2}+2-3d,\geq d]]$ quantum code, where $2\leq d\leq q$ is even;
  \item there exists a $q$-ary $[[2q^{2}-2,2q^{2}+1-3d,\geq d]]$ quantum code, where $2\leq d\leq q$ is odd.
\end{enumerate}
\end{theorem}

In Table~\ref{compa}, we list some quantum codes obtained from Theorem~\ref{MP7}. The table shows that our quantum codes have better parameters than the previous quantum codes available.
\begin{table}[h]
\begin{center}
\caption{Quantum Codes Comparison}
\begin{tabular}{|c|c|}
\hline
new quantum codes  &  quantum codes from \cite{E12}\\ \hline
$[[20,14,\geq3]]_{3}$ & $[[20,12,3]]_{3}$ \\ \hline
$[[48,28,\geq8]]_{5}$ & $[[48,26,8]]_{5}$ \\ \hline
$[[52,44,\geq4]]_{5}$ & $[[52,42,4]]_{5}$ \\ \hline
$[[52,40,\geq5]]_{5}$ & $[[52,38,5]]_{5}$ \\ \hline
$[[96,64,\geq12]]_{7}$ & $[[96,62,12]]_{7}$ \\ \hline
$[[100,92,\geq4]]_{7}$ & $[[100,92,3]]_{7}$ \\ \hline
$[[164,152,\geq5]]_{9}$ & $[[164,150,5]]_{9}$ \\ \hline
$[[164,156,\geq4]]_{9}$ & $[[164,154,4]]_{9}$ \\ \hline
\end{tabular}
\label{compa}
\end{center}
\end{table}

\end{spacing}
\end{document}